\documentclass[11pt]{article}

\usepackage{color}
\usepackage{graphicx}
\usepackage{pstricks,pst-node,pst-plot,pst-coil}
\usepackage{longtable}
\usepackage{pdflscape}
\usepackage{amsmath}
\usepackage{pspicture}
\usepackage{epsfig}

\usepackage{pifont,xspace}
\usepackage{graphics}
\usepackage{graphicx}      
\usepackage{amssymb}
\usepackage{epsf}
\usepackage{citesort}
\usepackage{wrapfig}
\newtheorem{Theorem}{Theorem}
\newtheorem{Lemma}[Theorem]{Lemma}

\newtheorem{remark}{Remark}

\newenvironment{proof}{{\noindent\bf Proof.\ }}{\hfill{\Pisymbol{pzd}{113}}\vspace{0.1in}}
\newenvironment{proof-sketch}{{\noindent\bf Sketch of proof.\ }}{\hfill{\Pisymbol{pzd}{113}}\vspace{0.1in}}

\newcommand{\TB}{\vspace{-0.1ex}}\newcommand{\TiE}{\setlength{\itemsep}{-1ex}}

\newcommand{\IE}{{\em i.e.}\xspace}
\newcommand{\EG}{{\em e.g.}\xspace}
\newcommand{\EA}{{\em et al.}\xspace}
\newcommand{\FI}[1]{Fig.~\ref{#1}}

\newcommand{\tx}{^{\rm th}}

\newcommand{\cS}{\mathcal{S}}
\newcommand{\cT}{\mathcal{T}}
\newcommand{\cU}{\mathcal{U}}

\newcommand{\cP}{\mathcal{P}}
\newcommand{\eps}{\varepsilon}
\newcommand{\ignore}[1]{}

\newcommand{\halmos}{\rule{1ex}{1.4ex}}
\newcommand{\qed}{\hfill \halmos} 
\newcommand{\be}[1]{\begin{equation}\label{#1}}
\newcommand{\ee}{\end{equation}}
\newcommand{\bi}{\begin{itemize}}
\newcommand{\ei}{\end{itemize}}
\newcommand{\ben}{\begin{enumerate}}
\newcommand{\een}{\end{enumerate}}

\newcommand{\cA}{\mathcal{A}}

\newcommand{\bl}[1]{\begin{Lema}\label{#1}}
\newcommand{\el}{\qed\end{Lema}}
\newcommand{\bt}[1]{\begin{Teo}\label{#1}}
\newcommand{\et}{\end{Teo}}
\newcommand{\epr}{\end{proof}}
\newcommand{\bpr}{\begin{proof}}
\newcommand{\eeqn}{\end{eqnarray*}}

\newcommand{\comment}[1]{}

\newcommand{\All}{{4-ALLELE}}
\newcommand{\TAll}{{2-ALLELE}}
\newcommand{\KAll}{{$k$-ALLELE}}

\newcommand{\bee}{\mathbf{e}}

\title{
On Approximating Four Covering and Packing Problems 
}

\author{
Mary Ashley\thanks{Supported by NSF grant IIS-0612044.} \\
Department of Biological Sciences \\
University of Illinois at Chicago \\
Chicago, IL 60607-7053 \\
Email: {\tt ashley@uic.edu} \\
\and
Tanya Berger-Wolf$\,\,\,^\dagger$ \\
Department of Computer Science \\
University of Illinois at Chicago \\
Chicago, IL 60607-7053 \\
Email: {\tt tanyabw@cs.uic.edu} \\
\and
Piotr Berman \\
Department of Computer Science \& Engineering \\
Pennsylvania State University \\
University Park, PA 16802 \\
Email: {\tt berman@cse.psu.edu} \\
\and
Wanpracha Chaovalitwongse$\,\,\,\dagger$ \\
Department of Industrial Engineering \\
Rutgers University \\
New Brunswick, NJ 08854 \\
Email: {\tt wchaoval@rci.rutgers.edu} \\
\and
Bhaskar DasGupta\thanks{Supported by NSF grants DBI-0543365, IIS-0612044, IIS-0346973
and DIMACS special focus on Computational and Mathematical Epidemiology.} \\
Department of Computer Science \\
University of Illinois at Chicago \\
Chicago, IL 60607-7053 \\
Email: {\tt dasgupta@cs.uic.edu} \\
\and
Ming-Yang Kao \\
Department of Electrical Engineering \& Computer Science \\
Northwestern University \\
Evanston, IL 60208 \\
Email: {\tt kao@cs.northwestern.edu} \\
}

\begin{document}

\maketitle

\begin{abstract}
In this paper, we consider approximability issues of 
the following four problems: 
{\em triangle packing}, 
{\em full sibling reconstruction},  
{\em maximum profit coverage} and 
{\em $2$-coverage}.
All of them are generalized or specialized
versions of set-cover and have applications in 
biology ranging from full-sibling reconstructions 
in wild populations
to biomolecular clusterings; however, as this paper shows, their 
approximability properties differ considerably.  Our inapproximability constant
for the triangle packing problem improves upon the
previous results in~\cite{CC06,CR02};
this is done by directly transforming the inapproximability gap of  
H\aa stad for the problem of maximizing the number
of satisfied equations for a set of equations over GF$(2)$~\cite{H97} 
and is interesting in its own right. 
Our approximability results on the full siblings reconstruction problems 
answers questions 
originally posed by Berger-Wolf \EA\ ~\cite{Berger-Wolf_etal.07,Berger-Wolf_etal.05} and 
our results on the maximum profit coverage problem 
provides almost matching upper and lower bounds 
on the approximation ratio, answering a question
posed by Hassin and Or~\cite{HO06}.
\end{abstract}

\section{Introduction}

We consider four combinatorial optimization problems motivated by four separate applications
in computational biology. Each of them concerns with packing or covering and
falls under a general framework of covering/packing as described below.  
In the general framework, we have a finite universe of elements
and a collection of sets contained in the universe.  Optional parameters
can be added to the problem statement to specify problems in this framework,
and in this paper we use the following (in different combinations):
non-negative weights for elements, non-negative weights of sets, a limit
on the number of sets that can be selected, the minimum number of selected
sets that contain an element, and a family of ``conflicts'', pairs of sets
such that at most one set from a conflict pair can be selected.
Our goal is to select a sub-collection of sets that
satisfies the constraints (like covering all nodes as required or not
containing conflict pairs) and
that optimizes an objective function which is linear in terms
of the weights of the sets and elements in our selection.
For example, both the minimum weight set-cover and
the maximum weight
coverage problem falls under the above framework.
We start out 
with the precise definitions of our problems and later describe their
motivations. 

\paragraph{Triangle Packing Problem (TP)~\cite{CR02,GRCCW98,HS89}}
We are given an undirected graph $G$.  A triangle is a cycle of 
$3$ nodes. The goal is to find (pack) a maximum number of
{\em node-disjoint} triangles in $G$. 

\paragraph{Full Sibling Reconstruction Problems (\KAll$_{n,\ell}$ for 
$k\in\{2,4\}$)~\cite{ashley-et-al-to-appear,Berger-Wolf_etal.05,Berger-Wolf_etal.07,Chaovalitwongse_etal.06,saad1,saad2}} 

Here the universe $\cU$ consists of $n$ elements. 
To partially motivate the problem, think of each element as an individual
in a wild population. 
Each element $p$ is
a sequence
$(p_1,p_2,\ldots,p_\ell)$ 
where each $p_j$ is a genetic trait ({\em locus}) and
is represented as an ordered pair $(p_{j,0},p_{j,1})$ of numbers
({\em alleles}) inherited from its parents.
We also use ${\bf p}_j$ to denote the set $\{p_{j,0},p_{j,1}\}$.
Certain sets of individuals can be full sibling, \IE having 
the {\em same} pair of parents under the 
Mendelian inheritance rule. 
These sets are specified in an {\em implicit} manner in the
following way. 
The Mendelian inheritance rule states that
an individual $p=(p_1,p_2,\ldots,p_\ell)$ can be a child of a pair of
{\em parents}, 
say father $q=(q_1,q_2,\ldots,q_\ell)$ and
mother $r=(r_1,r_2,\ldots,r_\ell)$, 
if for each $i\in\{1,\ldots,\ell\}$ we have
$p_{i,0}\in{\bf q}_i$ and
$p_{i,1}\in{\bf r}_i$, or
$p_{i,0}\in{\bf r}_i$ and
$p_{i,1}\in{\bf q}_i$; see Figure~\ref{allele-fig} for
a pictorial illustration.
This gives rise to two necessary
conditions for a set $\cA$ of elements to be 
full siblings.

\begin{wrapfigure}{r}{5in}
\epsfig{file=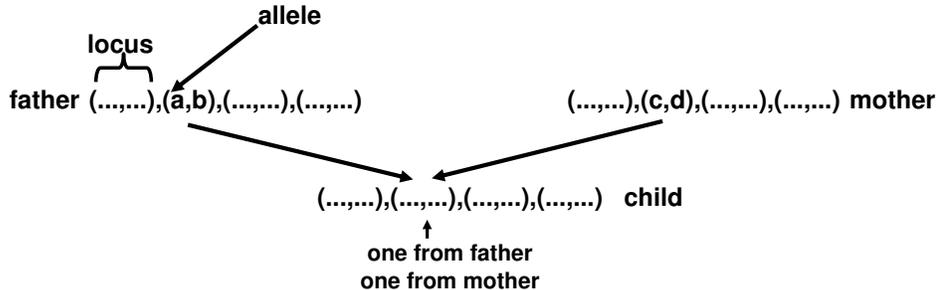,width=4in}
\caption{\label{allele-fig}
Illustration of the Mendelian inheritance rule.
} 
\end{wrapfigure}

Since each individual is generated by the same set of parents, each
having at most two distinct alleles in each locus, 
a set $\cA$ of elements can be full siblings if at most $4$ alleles occur
in each locus, i.e.,
$|\,\cup_{p\in\cA}\,{\bf p}_j\,|\leq 4$
for every $j\in\{1,2,\ldots,\ell\}$. 
Sets generated in this manner are said to satisfy the 
{\tt $4$-allele condition}. Notice that the
$4$-allele condition is not a sufficient condition for
individuals to be full siblings since it allows an individual
to inherit both its alleles from the same parent which violates
the Mendelian inheritance rule; nonetheless this condition is
used in practice since it is easy to check.

In a more precise way, the full sibling sets can be specified
via the {\tt $2$-allele condition}
described below.  In a full sibling set, we can reorder the alleles in each
locus of each individual in $\cA$ so that the first allele
always comes from the father
and second one
comes from the mother. Then, after such a reordering, 
a set $\cA$ of elements can be full siblings if at most $2$ alleles occur
in each coordinate of the locus. Formally, 
a set $\cA\subseteq \cU$ of elements 
satisfies the $2$-allele condition 
if and only if, 
for each 
$p\in\cA$ 
and 
each $j\in\{1,2,\ldots,\ell\}$, 
there exists a reordering
$\sigma_{p,j}=(\sigma_{p,j,1},\sigma_{p,j,2})
\in\{
(p_{j,0},p_{j,1}),
(p_{j,1},p_{j,0})\}
$ 
such that 
both 
$|\,\cup_{p\in\cA}\,\{\sigma_{p,j,1}\}\,|\leq 2$
and
$|\,\cup_{p\in\cA}\,\{\sigma_{p,j,2}\}\,|\leq 2$
for every $j\in\{1,2,\ldots,\ell\}$. 

With the Mendelian rules in mind,
the sets in the \KAll$_{n,\ell}$ problem are 
all possible sets of elements that satisfy
the $k$-allele condition for $k\in\{2,4\}$. 
The goal is then to find a collection of sets
that cover
the universe and the objective is to 
{\em minimize} the number of sets selected. 
As an example to illustrate the $k$-allele condition,
consider the $n=4$ elements (with $\ell=2$ loci) 
$p=(\{1,2\},\{5,5\})$,
$q=(\{3,4\},\{5,5\})$,
$r=(\{1,1\},\{5,5\})$ and 
$s=(\{5,5\},\{5,5\})$.
Then, there is no set containing all of $p,q,r$ and $s$ in 
either \All$_{4,2}$ or \TAll$_{4,2}$ because 
$|\,\{1,2\}\cup\{3,4\}\cup\{1,1\}\cup\{5,5\}\,|>4$, 
the set 
$\{p,q,r\}$ is contained in the instance of  
\All$_{4,2}$ but not in the instance of \TAll$_{4,2}$.

A natural parameter of interest in these problems is 
the maximum size (number of elements) $a$ of any set; 
we denote the corresponding problem by 
$a$-$k$-ALLELE$_{n,\ell}$ in some subsequent discussions.
One can make the following easy observations:
\begin{itemize}
\item
Both 2-\All$_{n,\ell}$ and 2-\TAll$_{n,\ell}$ are trivial
since any two elements always satisfy the $k$-allele condition
for $k\in\{2,4\}$.

\item
If $a$ is a constant, both 
$a$-\All$_{n,\ell}$ and $a$-\TAll$_{n,\ell}$
can be posed as a set-cover problem with a polynomially many sets with
the maximum set size being $a$ and thus 
have a $(1+\ln a)$-approximations (by using standard
algorithms for the set-cover problem~\cite{V01}). 

\item 
For general $a$, 
both 
$a$-\All$_{n,\ell}$ and $a$-\TAll$_{n,\ell}$
have a trivial $\left(\frac{a}{c}+\ln c\right)$-approximation for
any constant $c>0$ obtainable in the following manner.
For any integer constant $c>0$, it is trivial to find in polynomial time a set
of individuals that are full siblings for both  
\All$_{n,\ell}$ and \TAll$_{n,\ell}$, if such a set exists. 
Thus we can assume that for every induced instance of the
problem, either the maximum sibling group size is below $c$ and we can find such a group
of maximum size, or we can find a set of size $c$.  Obviously, we can assume
that if a sibling group can be used, we can use all its subsets too.
Consider an optimum solution, and make it disjoint.  We will distribute
the cost of an actual solution between the sets of the optimum.  When
a set with $b$ elements is selected, we remove each of its element
and charge the sets of the optimum $1/b$ for each removal.  It is easy
to see that a set with $a$ elements will get the sequence of charges
with values at most $(1/c,\ldots,1/c,1/(c-1),1/(c-2),\ldots,1)$
and these charges add to $\frac{a}{c}-1+\sum_{i=1}^{c}\frac{1}{i}$,
which in turn equals
$\frac{a}{c}+\sum_{i=2}^{c}\frac{1}{i}<\frac{a}{c}+\ln c$.
\end{itemize}

\paragraph{Maximum Profit Coverage Problem (MPC)~\cite{HO06}}
We have family of $m$ sets 
$\cS$ over a universe $\cU$ of $n$ elements.
For each $A\in\cS$ we have a non-negative {\em cost} $q_A$ and for each
$i\in\cU$ we have a non-negative {\em profit} $w_i$. 
We extend costs and profits to sets:
$q(\cP)=\sum_{S\in\cP}q_\cP$, and
$w(\cA)=\sum_{i\in\cA}w_i$.
For $\cP\subset\cS$ we define
the profit $c(\cP)=w(\cup_{A\in\cP}A)-q(\cP)$.
The goal is to find a subcollection of sets $\cP$ that maximizes $c(\cP)$.
A natural parameter for this problem is $a=\max_{A\in\cS} |A|$.
MPC admits a PTAS in the Euclidean space but otherwise its complexity
was unknown.

\paragraph{$2$-Coverage Problem}
Given $\cS$ and $\cU$ as in the MPC problem above and an integer $k>0$, 
a valid solution is
$\cP\subset\cS$ such that $|\cP|\leq k$; the goal is to {\em maximize} 
the number of elements 
that occur in {\em at least two} of the sets from $\cP$.
Another natural parameter of interest here is the 
{\em frequency} $f$, \IE, the maximum number of times any element
occurs in various sets. 

\subsection{Motivation}

In this section we discuss the motivations for the problems considered in this
paper. We discuss one motivation in details and
mention the remaining ones very briefly.

For wild populations, the growing development and application of
molecular markers provides new possibilities for 
the investigation of many fundamental biological phenomena, including
mating systems, selection and adaptation, kin selection, and dispersal
patterns. The power and potential of the genotypic information
obtained in these studies often rests in our ability to reconstruct
genealogical relationships among individuals. 
These relationships include parentage, full and half-sibships, and
higher order aspects of pedigrees~\cite{Blouin.03,Butler_etal.04,Jones_Ardren.03}. 
In our motivation we are only concerned with full sibling relationships
from single generation sample of microsatellite markers
Several methods for sibling reconstruction from 
microsatellite data have been proposed
~\cite{Almudevar_Field.99,Almudevar.03,Beyer_May.03,Konovalov_etal.04,Painter.97,Smith_etal.01,Thomas_Hill.02,Wang.04}.
Most of the currently available methods
use statistical likelihood models and
are inappropriate for wild populations. 
Recently, a fully combinatorial 
approach~\cite{ashley-et-al-to-appear,Berger-Wolf_etal.05,Berger-Wolf_etal.07,Chaovalitwongse_etal.06,saad1,saad2} 
to sibling reconstruction has been introduced.  This approach uses the simple
Mendelian inheritance rules to impose constraints on the genetic
content possibilities of a sibling group.  A formulation of the inferred
combinatorial constraints under the parsimony assumption of
constructing the smallest number of groups
of individuals that satisfy these constraints leads to the 
full sibling problems discussed in the paper. 
Both the $4$-allele and the $2$-allele constraints encode the above 
biological conditions for full siblings with varying strictness.
In this paper we 
study of computational complexity issues of these approaches.

MPC has applications in clustering
identification of molecules~\cite{HO06}.
The $2$-coverage problem has motivations in optimizing 
multiple spaced seeds for homology search (for relevant concepts, 
see \EG~\cite{XBLM07}). 
For application of TP to genome rearrangement problems, see~\cite{BP96,CR02}. 

\section{Several Useful Problems for Reductions}
\label{3maxcut-proven}

Several known problems were used for hardness results.
Below we list many of these problems
together with the known relevant results.  
Recall that a \emph{$(1+\varepsilon)$-approximate solution} (or
simply an $(1+\varepsilon)$-approximation) of a minimization (resp.
maximization) problem is a solution with an objective value no larger
(resp. no smaller) than $1+\varepsilon$ times (resp.  $(1+\eps)^{-1}$
times) the value of the optimum, and an algorithm achieving such a
solution is said to have an {\em approximation ratio} of at most
$1+\varepsilon$.
A problem is $r$-inapproximable 
under a certain complexity-theoretic
assumption means that the problem does not have a 
$r$-approximation unless the complexity-theoretic assumption is false.

\paragraph{$3$-LIN-$2$} 
We are given a set of linear equations modulo $2$ with $3$ variables per equation.
Our goal is to maximize the number of equations that
are satisfied with a certain value assignment to the variables.
A well-known result by H\aa stad~\cite{H97} shows the following result: 
for every $\eps<\frac{1}{2}$
it is NP-hard to differentiate between the instances that have at least
$(1-\eps)m$ satisfied equations 
from those that have at most 
$\left(\frac{1}{2}+\eps\right)m$
satisfied equations.

\paragraph{MAX-CUT on a $3$-regular graph ($3$-MAX-CUT)}
An instance is a $3$-regular graph, \IE, a graph $G=(V,E)$ where the degree of every vertex 
is exactly $3$ (and thus $|E|=\frac{3}{2}|V|$).
For a subset of vertices $V'\subseteq V$, define 
$score(V')$ to be the number of edges with exactly one endpoint in $V'$ and the
other endpoint in $V\setminus V'$.
The goal is then to find $V'\subseteq V$
such that $score(V')$
is maximized.
We will need the following inapproximability result for this problem
proved in~\cite{BK99}. 
For all sufficiently small constants $\eps>0$, 
it is impossible to decide, modulo RP$\neq$NP whether an instance $G$ of $3$-MAX-CUT with 
$|V|=336n$ vertices has a valid solution with a score below 
$(331-\eps)n$ or above $(332+\eps)n$.

\paragraph{Independent set problem for a $a$-regular graph (IS$_a$)}
A set of vertices are called independent if no two of them 
are connected by an edge. The goal is to find an independent set
of maximum cardinality when the input graph is $a$-regular, \IE, 
every vertex has degree $a$. 
It is well-known that 
this problem is
NP-hard for $a\geq 3$ and 
$a^c$-inapproximable
for general $a$ 
for some constant $0<c<1$ 
assuming
P$\neq$NP~\cite{AFWZ95,BS92,HSS03}.

\paragraph{Graph Coloring} 
The goal is to produce an assignment of colors to vertices 
of a given graph $G=(V,E)$ such that no two adjacent
vertices have the same color and 
the number of colors is {\em minimized}. 
Let $\Delta^\ast(G)$ denote 
the {\em maximum} number of independent vertices in a graph 
$G$
and $\chi^\ast(G)$ denote the minimum number of 
colors in a coloring of $G$.
The following inapproximability result is a 
straightforward extension of a hardness result 
known for coloring of $G$~\cite{FK96}:
for any two constants $0<\eps<\delta<1$, $\chi^\ast(G)$ cannot be 
approximated to within a factor of $|V|^{\eps}$ even if 
$\Delta^\ast(G)\leq |V|^{\delta}$ unless 
NP$\subseteq$ZPP.

\paragraph{Weighted set-packing} 
We have a collection of sets each with a non-negative weight over an universe.
Our goal is to select a collection of mutually disjoint sets of total maximum weight.
Let $a$ denote the maximum size of any set. 
For $a\leq 2$, weighted set-packing can be solved in polynomial time via 
maximum perfect matching in graphs.
For fixed $a>2$, Berman~\cite{B00} provided an approximation algorithm based on local improvements
for this problem produces an approximation ratio of $\frac{a+1}{2}+\eps$ for any constant $\eps>0$. 
When $a$ is {\em not} a constant, 
Algorithm $2$-IMP of Berman and Krysta~\cite{BerK03} provides 
an approximation ratio of $0.6454 a$ for any $a>4$. 

\paragraph{Densest Subgraph problem (DS)}
We are given a graph $G=(V,E)$ and a positive integer $0<k<|V|$.
The goal is to pick $k$ vertices such that the subgraph induced
by these vertices has the maximum average degree.
The densest subgraph problem is $(1+\eps)$-inapproximable
for some constant $\eps>0$ unless 
NP$\not\subseteq \cap_{\eps>0}$BPTIME$(2^{n^\eps})$~\cite{K04}.
A more general weighted version of DS admits a 
$O(|V|^{\frac{1}{3}-\eps})$-approximation for some constant $\eps>0$~\cite{FPK01}. 

\paragraph{Maximum coverage problem} 
This is the same as the $2$-coverage problem except that 
the number of elements that occur in at least {\em one}
of the selected sets is {\em maximized}.
Recall that $k$ is the number of sets that we are supposed to select
and $f$ is the frequency, \IE,
the maximum number of times any element
occurs in various sets. 
Let $\bee$ denote the base of natural logarithm.
It is known that the maximum coverage problem 
can be approximated
to within a ratio of 
$
1-\left(1-\frac{1}{k}\right)^k > 1-(1/\bee)
$
by a simple greedy algorithm~\cite{KMN99}
and approximation 
with ratio better than 
$
1-(1/\bee)
$
is not possible unless $P=NP$~\cite{F98}.
Obviously, the same lower bound carries over to $2$-coverage also
{\em for arbitrary} $f$.

\subsection{Our Results and Techniques}

The following table summarizes our results:


\begin{table}[h]
{\footnotesize
\hspace*{-0.5in}
\begin{tabular}{||cc||c|c|l||c||}
\hline
\multicolumn{2}{||c||}{Problem}& \multicolumn{3}{c||}{Lower Bound ($r$-inapproximability)}& Upper Bound \\ \cline{3-5}
         &                  & $r=$    &  assumption & reduction problem                & $r$-approximation for $r=$            \\
\hline
\hline
\multicolumn{2}{||c||}{Triangle Packing}                          & $(76/75)-\eps\approx 1.013$        & RP$\neq$NP & $3$-LIN-$2$ &      --- \\
\hline
\hline
\multicolumn{2}{||c||}{\{2,4\}-ALLELE$_{n,\ell}$}                                       &                                     &             &             & \\
$a=3$ & $\ell=O(n^3)$        & $(153/152)-\eps\approx 1.0065$     & RP$\neq$NP & Triangle Packing          &  --- \\
$a=3$ & any $\ell$         &  ---                                & ---         &             & $(7/6)+\eps\approx 1.166$ \\
$a=4$ & $\ell=2$             & $(6725/6724)-\eps\approx 1.00014$  & RP$\neq$NP  & $3$-MAX-CUT & --- \\
$a=4$ & any $\ell$         & ---                                 & ---         &             & $(3/2)+\eps$ \\
$a=n^\delta$ & $\ell=O(n^2)$ & $\Omega(n^\eps)\,\,\,\forall\eps<\delta$                    & ZPP$\neq$NP & graph coloring & $\eps n^\delta-\ln\eps$ $\,\,\forall$ constant $\eps$\\
\hline
\hline
\multicolumn{2}{||c||}{Maximum Profit Coverage}                                                 &                                     &             &             &  \\
\multicolumn{2}{||c||}{$a\leq 2$}                 &  ---                                & ---         & ---         &  polynomial-time \\
\multicolumn{2}{||c||}{$a\geq 3$}                 & NP-hard                             & ---         & $a$-regular indep. set      &  --- \\
\multicolumn{2}{||c||}{constant $a$}              & ---                                 & ---         & ---         & $0.5a+0.5+\eps$  \\
\multicolumn{2}{||c||}{any $a$}                   & $a^c$& P$\neq$NP   & $a$-regular indep. set      & $0.6454 a+\eps$  \\
\hline
\hline
\multicolumn{2}{||c||}{$2$-Coverage}                                                 &                                     &             &             &  \\
\multicolumn{2}{||c||}{$f=2$}                     & $1+\alpha$                          & NP$\not\subseteq \cap_{\eps>0}$BPTIME$(2^{n^\eps})$&  Densest Subgraph   & $O\left(m^{\frac{1}{3}-\beta}\right)$ \\
\multicolumn{2}{||c||}{any $f$}                   & ---                                 & ---                       & ---   & $O(\sqrt{m})$ \\
\hline
\end{tabular}
}
\caption{\small
Summary of results in this paper. By 
\{2,4\}-ALLELE$_{n,\ell}$ we mean that the results apply to {\bf both} 
\All$_{n,\ell}$ and \TAll$_{n,\ell}$. 
$0<\eps,\delta<1$ are {\bf any} two constants.
$\alpha$, $\beta$ and $c$ are {\bf specific} constants mentioned in~\cite{K04}, 
~\cite{FPK01}~and~\cite{HSS03}, respectively, but not explicitly calculated.
The parameters $a,\ell,f$ and $m$ 
are described in the definitions of the corresponding problems.
The $a^c$-inapproximability result for MPC holds 
{\bf even if} every set has weight $a-1$, every element has weight $1$, every set contains exactly $a$ elements and even if
we impose further restrictions such as 
each element is a point in some underlying metric space and each set correspond to a ball of
radius $\beta$ for some fixed specified $\beta$.
}
\end{table}

Brief descriptions of our techniques and comparisons with relevant previous results are
as follows.

\paragraph{Triangle Packing (TP)}
The lower bound is shown by a careful reduction 
from $3$-LIN-$2$ that {\em roughly} shows that, assuming RP$\neq NP$, it is hard to distinguish 
between instances of TP with profit (the number of disjoint triangles)
of at most $75k$ as opposed
to a profit of at least $76k$ for every $k$, thereby giving us an inapproximability
ratio of $\frac{76}{75}\approx 1.013$. Our inapproximability constant is larger
than the constant $\frac{95}{94}\approx 1.0106$ reported in~\cite{CC06} (assuming P$\neq$NP). 
A proof of Caprara and Rizzi~\cite{CR02} is yet earlier and it
implies a still worse inapproximability constant. 

\paragraph{\All$_{n,\ell}$ and \TAll$_{n,\ell}$} 
The inapproximability results for 
the {\em smallest} non-trivial value of $a$, 
namely $a=3$, and $\ell=O(n^3)$, are obtained 
by reducing TP to
instances in which the same sets satisfy $2$- and $4$-allele conditions and
each node of the initial graph (the TP instance)
is annotated with a sequence of loci so these sets coincide with triangles.
The $\left(\frac{7}{6}+\varepsilon\right)$-approximation for
any $\ell$ and any constant $\varepsilon>0$ 
is easily achieved using the results of
Hurkens and Schrijver~\cite{HS89}. 

The inapproximability results for 
the {\em second smallest} non-trivial values of $a$ and $\ell$, namely 
$a=4$ and $\ell=2$, are obtained by reducing $3$-MAX-CUT via 
an intermediate novel mapping of geometric nature.
The $\left(\frac{3}{2}+\eps\right)$-approximations
are achieved by using the result of
Berman and Krysta~\cite{BerK03}. 

The inapproximability result for 
$a=n^\delta$, namely {\em all sufficiently large} values of $a$, 
is obtained by reducing a suitable hard instance of the 
graph coloring problem.  

In general, for all the above reductions for 
\All$_{n,\ell}$ and \TAll$_{n,\ell}$
additional loci are used carefully to rule out 
possibilities that would violate the validity of our reductions. 

\paragraph{Maximum Profit Coverage (MPC)} 
The hardness reduction is from the IS$_a$ and the approximation algorithms 
are obtained via the weighted set-packing problem.
The $(0.6454 a+\eps)$-approximation for arbitrary $a$ is obtained via
a very careful polynomial-time dynamic programming implementation of the $2$-IMP
approach in Berman and Krysta~\cite{BerK03} that {\em implicitly} maintains subsets 
for possible candidates for improvement that {\em cannot be explicitly enumerated} 
due to their non-polynomial number. 

\paragraph{$2$-coverage} 
The inapproximability result and approximation algorithms for $f=2$ 
are obtained by identifying the problem with 
the DS problem. 
Note that the $1-(1/\bee)$-inapproximability result
for maximum coverage does not extend to $2$-coverage under the assumption of $f=2$.
For arbitrary $f$, 
we show a $O(\sqrt{m})$-approximation by taking the better of a direct greedy approach
and another greedy approach based on the maximum coverage problem. 
Note that a significantly better than $O(\sqrt[3]{m})$-approximation for
$2$-coverage would imply a better approximation for DS
than what is currently known.

\newcommand{\soT}{{\frak T}_{\rm sol}}
\newcommand{\inT}{{\frak T}_{\rm ins}}
\newcommand{\noT}{{\frak N}}
\newcommand{\bbP}{{\Bbb P}}
\newcommand{\bbQ}{{\Bbb Q}}
\newcommand{\bbT}{{\Bbb T}}

\section{Inapproximability Result for Triangle Packing}
\label{hastad}

The theorem below gives a 
$(76/75)-\eps\approx 1.0133$-inapproximability for TP.

\begin{Theorem}\label{TP-new}
Assume RP$\neq$NP. If $0<\eps<1/2$, there is no RP
algorithm that for each instance of TP with $228n$ nodes and 
a triangle packing of size at least $(76-\eps)n$ returns
a triangle packing of size at least $(75+\eps)n$.
\end{Theorem}

\begin{proof}
For convenience to readers, we first describe the plan of the proof, then
an informal overview of the calculations and finally the details of
each component of the proof.

\noindent
{\bf Plan of the proof}.
As stated before, the following result was obtained 
by H\aa stad in ~\cite{H97}.
Let $L$ be any language in NP. Then, an instance $x$ of $L$ 
can be translated in polynomial time to an instance of $3$-LIN-$2$ with $2n$ equations 
such that, for any constant $0<\eps<\frac{1}{2}$, the following holds:
\begin{itemize}
\item
if $x\in L$, then we can satisfy at least $(2-\eps)n$ equations, and, 

\item
if $x\not\in L$, then we can satisfy at most $(1+\eps)n$ equations.
\end{itemize}
The above result therefore provides an $(2-\eps)$-inapproximability 
of $3$-LIN-$2$ for any small constant $\eps>0$, assuming P$\neq$NP.

Our randomized schema to prove the desired inapproximability result
modulo RP$\neq$ NP is as follows.
Our randomized reduction uses the following polynomial-time transformations
that we will devise:
\begin{description}
\item[(A)]
First, we have a randomized ``instance transformation'' $\inT$ that
maps an instance $S$ of $3$-LIN-$2$
with $2n$ equations into a graph $G\equiv\inT(S)$ with $228nm_S$ nodes ($m_S<n$
is a small integer related to the size of $S$).  The algorithm of $\inT$ is
randomized and the output is random. 
{\bf A crucial property of this
transformation is that  
with probability {\em at least} $1/2$ the output is {\em correct}, \IE,
the corresponding instance graph $G$ will satisfy the subsequent requirement in (C) below}.

\item[(B)]
Second, we have a (deterministic) ``solution transformation'' $\soT$ that maps a solution, say $s$, of 
the instance $S$ of $3$-LIN-$2$ with $2n$ equations
to a solution $\soT(s,G)$ of the triangle packing problem in the above-mentioned graph $G$.
Our transformation will satisfy the following properties:
\begin{description}
\item[(a)]
If $s$ satisfies $2n-\ell$ equations of $S$ then
$\soT(s,G)$ has $(76n-\ell)m_S$ triangles in $G$ (and $3m_S\ell$ nodes not covered
by the triangles).
In particular, note that this implies that, 
\begin{itemize}
\item
if we satisfy $(2-\eps)n$ equations of $S$ then 
$\soT(s,G)$ has $(76-\eps)nm_S$ triangles in $G$, and, 

\item
if we satisfy $(1+\eps)n$ equations of $S$ then 
$\soT(s,G)$ has $(75+\eps)nm_S$ triangles in $G$.
\end{itemize}

\item[(b)]
We can find $s$ in polynomial-time if we are given $\soT(s,G)$.
\end{description}

\item[(C)]
Third, we have ``solution normalization'' transformation $\noT$ maps a 
triangle packing $\bbP$ in the graph 
$G$ into another triangle packing $\noT(\bbP,G)$ in the graph $G$ 
which is of the form $\soT(s,G)$ for some
solution $s$ of the instance $S$ of $3$-LIN-$2$.  
If $G$ is a ``correct output'' of $\inT(S)$ then $|\noT(\bbP,G)|\ge|\bbP|$, \IE,
normalization does not decrease the number of triangles in the solution.
\end{description}
Given the above transformation, the overall approach in our proof is 
as follows.
Suppose that we have a polynomial-time randomized algorithm $\frak A$ that with
probability at least $1/2$ finds triangle packing of size larger than
$(75+\eps)/(76-\eps)$ times the optimum (assuming that one exists).
Then, we can use $\frak A$ to devise an RP algorithm for any language in 
NP in the following manner: 
\begin{description}
\item[(a)]
We start with an instance $x$ of a language $L\in$NP. 
Using the proof of H\aa stad in ~\cite{H97}
we translate $x$ in polynomial time to the corresponding instance of 
$S$ $3$-LIN-$2$ with $2n$ equations. 

\item[(b)]
We compute $G=\inT(S)$.

\item[(c)]
We compute the triangle packing solution $\bbP={\frak A}(G)$.

\item[(d)] 
We compute a new triangle packing solution 
$\bbQ=\noT(\bbP)$
using the normalization transformation $\noT$.

\item[(e)]
if $|\bbQ|<|\bbP|$ then we repeat steps (b)-(d) up to a
polynomial number of times.

\item[(f)]
if $|\bbQ|<|\bbP|$ in some execution of Step~(e) 
then 
we find the solution $s$ of $S$ that corresponds to $\bbQ$.
If $s$ satisfies strictly more than $(1-\eps)n$ equations 
then we declare $x\in L$. In all other cases we declare $x\not\in L$.

One can now see that we are always correct if $x\not\in L$ and we
are correct with probability at least $1/2$ if $x\in L$.
\end{description}

\noindent
{\bf 
An informal overview of the calculations involved in instance transformation
$\inT$.}
The transformation $\inT$ from an instance $S$ of $3$-LIN-$2$ to 
an instance (graph) $G$ of triangle packing goes through the
following stages.
In $S$ we have a system of $2n$ equations modulo $2$, with $3$ literals per equation,
and we can satisfy either at most $(\frac{1}{2}+\eps)$ fraction of the equations 
or at least $(1-\eps)$ fraction of the equations. 

First, we replicate each equation some (polynomial) $m$ times.  This is to increase
the minimum number of occurrences of each variable such that 
the ``consistency gadgets'' for occurrences
will be correct -- the correctness of these gadgets is proved ``in the limit'',\IE,  
starting from a certain size.
This {\bf does not change} the fraction of equations in the system that can
be simultaneously satisfied, which is either $1-\eps$ or $\frac{1}{2}+\eps$.

Denote by $\neg x$ the negation of the variable or constant $x$ modulo $2$, \IE, 
$\neg x=x+1\pmod 2$. Then, any equation can have two ``normal'' forms, 
namely,
\[
x+y+z=b\pmod 2
\]
\[
\neg x+\neg y+\neg z=\neg b\pmod 2
\]
We now replace each equation with such a pair. Again, this does not change
the proportion of the equations that can be simultaneously satisfied.
Our reductions and instance/solution transformations will ensure that each
variable $\neg x$ receives a value which is the negation of the value
received by variable $x$.
The above replications together account for the constant $m_S$ 
mentioned in item~(A) of the plan of the proof.
In other words, after these replications, we have 
$nm_S$ variables.

Now, our system of equations have some nice properties:
\begin{itemize}
\item
roughly, for each two equations, both can be satisfied or one;

\item
same number of negated and non-negated literals;

\item
same number of equations ``$=0\pmod 2$'' and ``$=1\pmod 2$''

\item
assured minimum number of occurrences of each variable.
\end{itemize}
Now, we show our calculation on a normal pair of equations
as discussed in the replication method above.
\begin{itemize}
\item 
We have $6$ occurrences of literals.
We will design a ``triplicate gadget'' for each, in which each occurrence is
represented as $3$ nodes called {\em literal nodes}, thus we have a total of 
$18$ literal nodes.
We will design a single gadget for each ``$=0$'' equation that has $6$ other nodes, and a gadget 
for each ``$=1\pmod 2$'' equation that has 
$4$ other nodes. 
Thus, we have $10$ extra nodes for each normal pair of equations,
which makes $30$ extra nodes in a ``triplicate gadget''.

\item
For each $18$ literal node, we will have a part of a consistency gadget
in which we have $7$ triangles that make a sequence of overlaps.
Together, these triangles would have $21$ nodes,
but one of these node is the literal node, and of the other $20$,
each is shared with another triangle, so they are really $10$ distinct
nodes.
For a pair of triplicate gadgets, we have $10\times 18=180$ of the nodes
of consistency gadgets.
\item
Thus, together, we have $(180+30+18)nm_S=228nm_S$ nodes.

\item
Roughly, the two cases of triangle packing (ignoring the $\eps$ factors and
so forth) are as follows.
When both equations in the normal pair are satisfied, we cover them completely with $76$ triangles,
and when one equation fails, we will loose one triangle thereby covering with $75$ triangles.
\end{itemize}

\noindent
{\bf The outline of the instance translation.}
Given $S$, a system of $2n$
equations with $3$ variables per equations, we proceed as follows.

\begin{enumerate}
\item
\label{replia}
We replicate each equation six times, three times as a simple copy,
$x+y+z=b\mod~2$ and three times as
$\bar{x}+\bar{y}+\bar{z}=\bar{b}\mod~2$.  Having the same number of
literals $x$ as $\bar{x}$ helps in point \ref{gadgetc}, and having each
equation copied three times helps in point \ref{gadgete}.
\item
\label{replib}
We replicate the equations in $S$ $m$ times for a sufficiently 
(polynomially) large $m$ such each variable occurs sufficiently (polynomially)
many times.
The construction in point~\ref{gadgetc}
is faulty with probability $O(c^{m'})$ for some $c<1$ when $m'$ the number of
occurrences of a variable.
\item
For each literal (occurrence of a variable or its negation in an equation) we
create a separate
node.  From now on, {\em literal} will mean such a node.
\item
\label{gadgete}
We replace three copies of equation $e$ with an {\em equation gadget} $B_e$
that contains nine literals of $e$ (three, each in three copies) as well
as other nodes.
\item
\label{gadgetc}
For each variable $x$ we create {\em consistency gadget} $C_x$ that all 
the literals of $x$, as well as other nodes.
\end{enumerate}

\noindent{\bf Constructing consistency gadget $C_x$.}

The problem of triangle packing can be mapped into the independent
set problem in the following manner: starting from a graph $(V,E)$
we create a graph $(V',E')$, where $V'$ is the set of triangles in $E$,
and $\{t,t'\}\in E'$ if triangles $t$ and $t'$ share a node.

If graph $G'$ is cubic, i.e. each node has degree 3, we can have the reverse
transformation: from $(V',E')$ to $(V,E)$; $V=E'$, and $\{e,e'\}\in E$ if
$e$ and $e'$ are incident to the same node.  In this case, a node $u\in V$
with neighbors $v_i$, $i=0,1,2$, is transformed into nodes $\{u,v_i\}$,
$i=0,1,2\}$ and those three nodes for a triangle.  A pair of such triangles
is node-disjoint if the original nodes were not adjacent.

This point of view is not helpful in the construction of equation gadgets
because we obtained smaller gadgets than those that would correspond to
fragments of cubic graphs.  However, our consistency gadget are obtained
by such a transformation.

In particular, we will use a gadget, called an {\em amplifier}, introduced by
Berman and Karpinski \cite{BK99} 
in the context of maximum cut problem (see also
J. Chleb\'{i}kov\'{a} and M. Chleb\'{i}k \cite{CC06}).

Assume that we construct $G_x$ for a variable with $2k$ occurrences ($k$
simple, $k$ negated).  The respective amplifier can be defined as the
graph $(V^a,E^a)$ where
$V^a=\{u_0,\ldots,u_{14k-1}\}$,  This graph is bipartite, all edges are between
even nodes and odd nodes; we will refer to odd and even nodes as white and
black.  There are two classes of edges, the first forms a ring,
$\{u_i.u_{i+1~{\bf mod}~14k}\}$, the second forms a random matching between
white (even) and black (odd) nodes whose indices {\bf are not divisible by 7}.
Nodes with indices divisible by 7 are called
{\em contacts}, each of these nodes belongs also to an equation gadgets.

We wish a solutions -- a $U\subset V$ of nodes -- to be consistent within
consistency gadgets.  Equation gadgets ``see'' only the contacts.
Set $U$ is consistent within our gadget if either $U$ contains all black
contacts and none of the white ones, or vice versa.  If we have
an inconsistent solution, we replace it with the choice ``all white'' or
``all black'' that requires fewer changes of membership among the contacts.
Here is the key property (that holds with the probability that converges
to 1 as $k\rightarrow\infty)$):

\noindent
if $U\subset V^a$ contains $i\le k$ contacts of one color (the minority)
and at least as many nodes of another (the majority color), then at
least $i$ edges of $E^a$ do not belong to the cut of $U$.

The use of this property is that when we normalize a solution to coincide,
all contacts of $G_x$ should correspond to a single value assigned to $x$;
we can achieve it by altering the solution to coincide the color that contains
more contacts.  If the normalization changes the membership of $i$
contacts of $x$, we gain $i$ units of the objective function --- edges of the
cut --- within the gadget.
Presumably
the size of the cut decreases within equation gadgets. but the decrease
is bounded by $i$, the number of contacts that changed the membership.

Now we have to translate this usage of the amplifier to the independent set
problem.  In a bipartite cubic graph with $14k$ nodes, an independent set $S$
has cut $3|S|$, and if we have $3i$ edges not in the cut, then
$|S|=7k-i$.  Thus the same amplifier construction can be used for
independent set problem.

\noindent
if $U\subset V^a$ contains $i\le k$ contacts of its minority color, then at
least $i$ edges of $E^a$ are not covered by $U$.

Then we can translate the amplifier into a part of triangle packing
as shown in \FI{ampliftrans}, and the property can be rephrase by having $i$
nodes not covered by the solution triangle packing within an equation
gadget if $i$ contacts are covered in a minority manner (if the majority
of contacts covered by a solution is black, black is the majority color
and inconsistent consistent contact are black contacts that do not
belong, as well as white contacts that do belong.

\begin{figure}
\centerline{
\epsfig{file=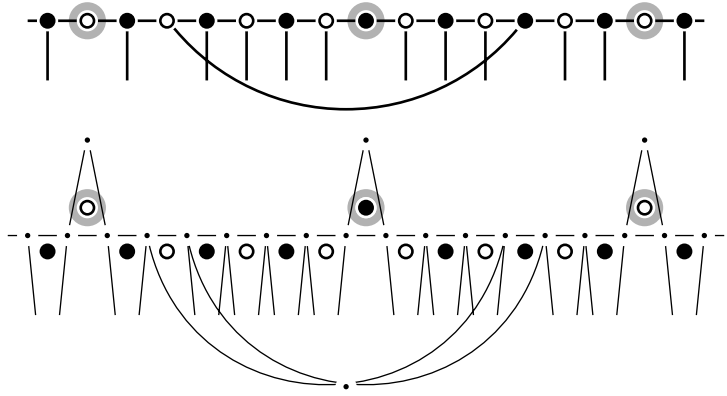}
}
\caption{\label{ampliftrans}
A fragment of an amplifier and its translation into a fragment
of our equation gadget.
Contacts are indicated by a gray ``halo''.
Note that after translation, each original contact node becomes
a contact triangle.  Each contact triangle contains a contact node
(in the diagram, on top).
If we choose white triangles, then contact nodes of the black triangles
are not covered within the gadget, and vice versa when we choose all black
triangles.
}
\end{figure}

\noindent{\bf How equation gadget $B_e$ works.}
Equations were replicated so they can be grouped into triples of identical
equations.  We create gadgets for equations and then, for each group of three,
we connect identical gadgets by
providing triangles that cover one node in each of them.

For such a group of copies of equation $e$, let $B^i_e$, $i=0,1,2$,
be an individual gadget and $B_e$ the combined one.

Thus we can describe a triple gadget by describing an individual gadget,
$B^i_e=(V^i_e,E^i_e)$ and specifying set $S^i_e$ of nodes that are
connected to their copies in other individual gadget.  From the point
of view of an individual gadget, nodes in $S^i_e$ can be covered
separately.

Assume that $e\equiv x'+y'+z'=b~\mod~2$ where $x'$ is a literal of $x$ ($x$
or $\bar{x})$.  An individual gadget contains these three literals.

The property of an individual gadget $B^i_e$ is that $V^i_e$ can have all nodes
covered by a triangle packing and $S^i_e$ if only if the literals are covered
consistently with values that make $e$ satisfied.  For example, if 
$e\equiv x+y+z=0~\mod~2$, and none (or exactly two) of the three literals
contained in $V^i_e$ is covered by triangles contained in $C_x\cup C_y\cup C_z$.
The property of the combined gadget is that if the literals are covered
consistently (e.g. either all $x'$ are covered by triangles contained in $C_x$
or none), then either they are covered consistently with values that satisfy
$e$ and we can cover entire $V_e=V^0_e\cup V^1_e\cup V^2_e$, or the literals
are covered consistently with values that do not satisfy $e$ and we can cover
$V_e$ except for three nodes (one exception in each $V^i_e$).

\noindent{\bf Properties of gadgets imply correct normalization.}
So far, we described $\bbQ=\noT(\bbP)$ only partially, namely how to select
triangles contained in consistency gadget $G_x$ (white or black, corresponding
to assigning 0 or 1 to $x$).  If the normalization
change the way $i$ contacts are covered, then within $G_x$ we cover
all nodes with the triangle, while before we did not cover $i$ of them.
Thus we can pass to each
``minority case'' a permission not to cover one node.

Now consider a combined equation gadget.
If the majority cases satisfy the equation, after the normalization we
cover all nodes of the equation gadget.  Otherwise each individual gadget
either contained a minority case literal and will receive a permission
not to cover a node, or it had all majority cases and thus at least
one uncovered node.  Thus to maintain the number of covered nodes
it suffices to cover the nodes in the gadget with three exceptions.

\begin{wrapfigure}{r}{4in}
\centerline{\begin{picture}(300,110)(0,-10)
\put(10,10){\epsfig{file=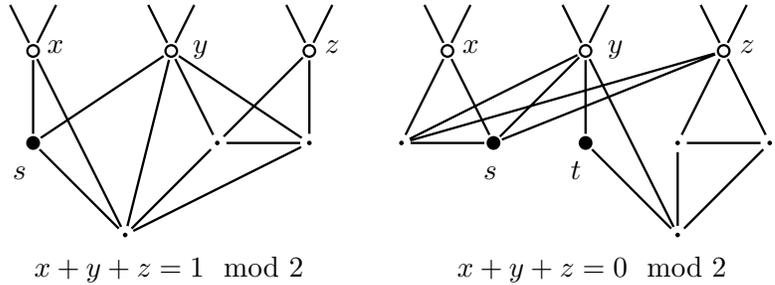,width=4in}}
\put(20,-5){$x+y+z=1\mod 2$}
\put(25,80){$x$}
\put(80,80){$y$}
\put(130,80){$z$}
\put(12,32){$s$}
\put(180,-5){$x+y+z=0\mod 2$}
\put(182,80){$x$}
\put(237,80){$y$}
\put(287,80){$z$}
\put(190,32){$s$}
\put(223,32){$t$}
\end{picture}}
\caption{\label{triang} Equation gadgets, used in three copies.  Thick dots are
nodes connected with other copies (self-sufficient), empty circles are
literals, nodes shared with consistency gadgets of variables.}
\end{wrapfigure}

\noindent
{\bf Construction of $B^i_e$}

Consider equation $e\equiv x+y+z=0 \mod~2$.  Node set $V^i_e$
consists of three literals (one copy of $x,y,z$), two self-sufficient
nodes $S^i_e=\{s^i,t^i\}$ and four other nodes.

If $x,y,z$ are false, this is coded by a solution in which
none is already covered
by triangles from their consistency gadgets,
we cover the nine nodes of $B^i_e$ with three
triangles.  If exactly two are already covered, we
cover the uncovered literal, $s^i$ and ``four other nodes'' with two triangles.

If exactly one of the literals true (already covered), we would have to cover
eight
nodes.  This could be done only with two triangles and two self-sufficient
nodes; however the triangles disjoint with $S^i_e$ all overlap,
so the best we can do is to use one such triangle, one triangle that
contains $s^i$ and $t^i$, leaving one non-self-sufficient node uncovered.

If three literals are true, we would have to cover six nodes, this
could be done only with two triangles, but there is only one triangle
that does not contain literals, so the best we can do is to use this
triangle, as well as $S^i_e$,
leaving one of the ``other nodes'' uncovered.

Now consider equation $e\equiv x+y+z=1 \mod~2$.  Sub-gadget $B^i_e$
contains $x^i,y^i,z^i$, self-sufficient node $s^i$ and three other nodes.

If exactly one of the literals is true, we have to cover
six nodes, which we
can do with two triangles.  If all literals are true, we have to cover
4 nodes, which we do using a triangle that is disjoint with $s^i$, as
well as $s^i$.

If no literal is true, we would have to cover 7 nodes, this could be
done only with two triangles and $s^i$, but all triangles that do not
contain $s^i$ overlap.  If two literals are true, we would have to
cover 5 nodes, impossible.  But if we pretend that one more literal
is covered we can cover all other nodes, so when the equation is false
we leave one non-self-sufficient node uncovered.

\noindent{\bf The property of the combined gadget}
It is easy to see that when the literals are consistent we can
cover each individual gadget in the same way, so when any nodes remain
uncovered they form triples of corresponding self-sufficient nodes
and thus they are covered by the triangles that connected individual gadgets.
\vspace{-4ex}

\end{proof}

\section{Approximability for \All$_{n,\ell}$ and \TAll$_{n,\ell}$ for $a=3$} 

\begin{Theorem}\label{basic-approx}
Both \All$_{n,\ell}$ and \TAll$_{n,\ell}$ 
are $((153/152)-\eps)$-inapproximable 
even if $a=3$ assuming RP$\neq$NP and 
(for any $\ell$) admit
$((7/6)+\varepsilon)$-approximation for
any constant $\varepsilon>0$.
\end{Theorem}

\begin{proof}
We reduce the {\em Triangle Packing} (TP) problem to our problem.
We will use the inapproximability result for TP as described
in Section~\ref{hastad}. 

To treat both \All$_{n,\ell}$ and \TAll$_{n,\ell}$ in an unified
framework in our reduction, it is convenient to introduce the 
$2$-label cover problem. 
The inputs are the same as in \All$_{n,\ell}$ or \TAll$_{n,\ell}$
except that each locus has just one value (label) and 
a set of individuals are full siblings if on every locus they have 
at most $2$ values. 
Thus, each individual can be thought of as an ordered sequence 
of labels.
An instance of the $2$-label cover problem can be translated 
to an instance of our problem by replacing each label in each 
locus in the following manner:
\begin{itemize}
\item
for \All$_{n,\ell}$, the label  
value $v$ is replaced by the pair 
$(v,v')$ where $v'$ is a new symbol; 

\item
for \TAll$_{n,\ell}$ the value $v$ is replaced by the pair $(v,v)$.
\end{itemize}
We will reduce an instance of 
TP to the $2$-label cover problem by introducing an 
individual for every node of the graph $G$ with $n$ nodes and 
providing label sequences for each node (individual) such
that:
\begin{description}
\item[$(\star)$] 
three individuals corresponding to a triangle of $G$ have
at most two values on every locus, and 

\item[$(\star\star)$] 
three individuals that do not correspond to a triangle of $G$ have
three values on some locus. 
\end{description}
Note that, since any pair of individuals can be full siblings, 
the above properties imply that TP has a solution with $t$
triangles if and only if the $2$-label cover can be covered with
$\frac{n-t}{2}$ sibling groups.
Thus, Theorem~\ref{TP-new} implies that it is NP-hard to decide
on instances of $228k$ individuals 
whether the number of full sibling groups
is above $(228-76+\eps)k/2$ or below $(228-75-\eps)k/2$, thereby giving
$(153/152)-\eps\approx (1.0064-\eps)$-inapproximability.

The index of a locus, which we call the coordinate, 
is defined by:
\begin{description}
\item[(a)] 
an ``origin'' node $a$, and 
\item[(b)] 
{\em optionally}, a certain edge $e$.
\end{description}
Thus, we will have at most $O(|V|\cdot |E|)$ loci. 
The respective label of a node $v$ at this coordinate 
is the distance from $a$ to $v$, assuming
every edge except $e$ has length $1$ while $e$ has length $0$. 
Let dist$(u,v)$ denote the distance between nodes $u$ and $v$. 

It is easy to see that Property~$(\mathbf{\star})$ holds.
Consider a triangle $\{u,v,w\}$ and assume that $u$ has the minimum label
value of $L$, \IE, it is the nearest with respect to the origin node that defined 
this locus. Then labels of $v$ and $w$ are at least
$L$ and at most $L+1$, hence we have at most two labels.

It is a bit more involved to verify Property~$(\mathbf{\star\star})$.
Consider a non-triangle $\{u,v,w\}$ in a labeling defined by
$u$ (with no edge). $u$ has label $0$ and $v,w$ have positive labels which
may be equal: if not, we are done; if yes, let $L=$dist$(u,v)=$dist$(u,w)$. 

Consider the two shortest paths from $u$ to $v$ and $w$,
respectively, such
that they share a maximally long initial part; so for some node $x$
dist$(u,v)=$dist$(u,x)+$dist$(x,v)$, \\ 
dist$(u,w)=$dist$(u,x)+$dist$(x,w)$ 
and the shortest paths from $x$ to $v$ and $w$ have to be disjoint.
Let $\{x,y\}$ be an edge on a shortest path from $x$ to $v$ and now set
its length to $0$.

First, observe that 
dist$(y,w)\geq$dist$(x,w)$, since otherwise 
dist$(y,w)\leq$dist$(x,w)-1$, 
dist(u,v) = dist(u,x)+dist(x,y)+dist(y,v) and also
dist$(u,w)=$dist$(u,x)+$dist$(x,y)+$dist$(y,w)$ and we found a longer
common prefix of shortest paths from $u$ to $v$ and $w$.

Now when we shrink $e=\{x,y\}$ by setting its length to zero, 
the labels of $u$ and $w$ are unchanged and
the label of $v$ drops by $1$; we have only two labels only if
the labels of $u$, $v$ and $w$ are $0$, $1$ and $1$, respectively, 
which implies that $\{u,v\}$ and $\{u,w\}$ are edges.

In this case we label nodes by distances from $v$; $v$ gets $0$, $u$ gets
$1$, if $w$ also gets $1$ then we have edges $\{u,v\}$, $\{u,w\}$ and now
we witnessed $\{v,w\}$, hence $\{u,v,w\}$ is a triangle.

This completes the hardness reduction.

On the algorithmic side, 
suppose that an optimal solution for either version of the
sibling problem on $n$ individuals involve $a$ triples and $b$ pairs 
of individuals (and, thus, $3a+2b$). 
Hurkens and Schrijver~\cite{HS89} 
have a schema that approximates triangle packing within a ratio of
$1.5+\eps$ for any constant $\eps>0$.
We can use this algorithm to get at least $(2a/3)-\eps$ triples.
We can cover the remaining 
$n-(2a-3\eps)=a+2b+3\eps$ 
elements by pairs.
Thus, we use at most
$(2a/3)-\eps+(a/2)+b+(3/2)\eps=(7a/6)+b+(\eps/2)$
which is within a factor of $(7/6)+\eps$ of 
$a+b$. 
\end{proof}

\section{Approximability of \All$_{n,\ell}$ and \TAll$_{n,\ell}$ for $a=4$} 

\begin{Theorem}\label{main2}
For $a=4$, 
both \All$_{n,\ell}$ and \TAll$_{n,\ell}$ 
are $((6725/6724)-\eps)$-inapproximable 
even if $\ell=2$ 
assuming RP$\neq$NP and 
(for any $\ell$) admit
$((3/2)+\varepsilon)$-approximation 
for any constant $\varepsilon>0$.
\end{Theorem}

\begin{proof}
We will prove the result for 
\TAll$_{n,\ell}$ only; a proof for
\All$_{n,\ell}$ can be obtained by an easy modification 
of the above proof.
We will prove the result by showing that, for any constant $\eps>0$,
\TAll$_{n,\ell}$ cannot be approximated to within a ratio of 
$\frac{6725}{6724}-\eps$ unless RP$=$NP. 

We will reduce an
instance $G=(V,E)$ of $3$-MAX-CUT
to \TAll$_{n,\ell}$ and use the previously proved result on
$3$-MAX-CUT as stated in Section~\ref{3maxcut-proven}. 
For notational simplicity, let $m=|E|$. 
We will provide a reduction from an instance $G=(V,E)$ of $3$-MAX-CUT 
with $336n$ vertices 
to an instance of \All$_{10m,\ell}$ with $\ell=2$. 
The reduction will satisfy the following properties:
\begin{description} 
\item[(i)]
a solution of $3$-MAX-CUT with a score of $x$ 
correspond to 
a solution of 
\TAll$_{24m,2}$ 
with 
$14m-x$ sibling groups;

\item[(ii)]
a solution of 
\TAll$_{24m,2}$ 
with $z$ sibling groups 
can be transformed in polynomial time 
to another solution of 
\TAll$_{24m,2}$ 
with $14m-y\leq z$ sibling groups (for some positive integer $y$)
such that this solution correspond to a solution of $3$-MAX-CUT 
with a score of $y$.
\end{description} 
Note that this provides the required gap in approximability. Indeed,
observe that (with $m=336\times\frac{3}{2}\times n=504n$) 
$3$-MAX-CUT has a solution of score below 
$(331-\eps)n$ if and only if 
\TAll$_{24m,2}$ 
has a solution with at least 
$14\times 504n-(331-\eps)n=(6725+\eps)n$ 
sibling groups
and conversely 
$3$-MAX-CUT has a solution of score above 
$(332+\eps)n$ if and only if 
\TAll$_{24m,2}$ 
has a solution with at most 
$14\times 504n-(332+\eps)n=(6724-\eps)n$ 
sibling groups;
thereby the inapproximability gap is 
$\frac{6725}{6724}-\eps$. 

When we look at {\em one locus only}, a set of full siblings can have a very
limited set of values for alleles. Consider first the case in which every individual
has two different elements (alleles) at this locus. We can then view each individual
$\{u,v\}$ 
as an edge in an undirected graph with the two elements $u$ and $v$ representing 
two nodes in the graph. Three edges (individuals) can be full siblings if they
form a path or a cycle; if they do not form a connected graph their
union has more than $4$ elements, and if they are of the form
$\{u,v\},\{u,w\},\{u,x\}$ then also they violate the $2$-allele condition.
Four edges can be full siblings if they form a cycle since they must have only
$4$ nodes and $3$ edges incident on the same node violate the $2$-allele
condition.
The other members in a full sibling group 
for an individual $\{u,u\}$ can be subsets of
either 
$\{\,\{u,v\},\{v,v\}\,\}$ or
$\{\,\{u,v\},\{u,w\},\{v,w\}\,\}$.
In our reduction cycles of length $3$ will not exist, so full siblings sets
of size larger than two will be
paths of $3$ edges, cycles of $4$ edges and triples of the
form $\{u,u\},\{u,v\},\{v,v\}$.
For the purpose of the reduction, it would be more convenient to
reformulate the properties~{\bf (i)} and {\bf (ii)} of the reduction 
described above by the following obviously equivalent properties:
\begin{description} 
\item[(i')]
a solution of $3$-MAX-CUT with a score of $m-x$ 
correspond to 
a solution of 
\TAll$_{24m,2}$ 
with 
$13m+x$ sibling groups;

\item[(ii')]
a solution of 
\TAll$_{24m,2}$ 
with $z$ sibling groups 
can be transformed in polynomial time 
to another solution of 
\TAll$_{24m,2}$ 
with $13m+y\leq z$ sibling groups (for some positive integer $y$)
such that this solution correspond to a solution of $3$-MAX-CUT 
with a score of $m-y$.
\end{description} 
We now describe our reduction. 
We are given a cubic graph $G$
with $2n$ nodes (and thus with $m=3n$ edges) and we will construct
an instance $J$ of 
\TAll$_{24m,2}$. 
We replace each node $u$ of $G$ with a gadget $G_u$
that consists of
$36$ individuals (see Figure~\ref{node-edge-gadget}). 
Our individuals have two loci. According to
the first locus, individuals are edges in a 4-regular graph.
Gadget $G_u$ is a $3\times 12$ grid.  The
rows are closed to form rings of 12 edges, and every fourth column
is similarly closed to form a ring on 3 edges.  This leaves
$6$ connected groups of $3$ nodes each with $3$ neighbors 
only (\EG, the second, third and fourth node from left on the first
row is one such group); these
groups are connected to similar groups in other gadgets.
A connection between two gadgets consists of two $2\times 3$ grids;
for each grid the two rows come from two above-mentioned
groups of nodes, one from each gadget. 

\begin{figure}[h]
\centerline{\epsfig{file=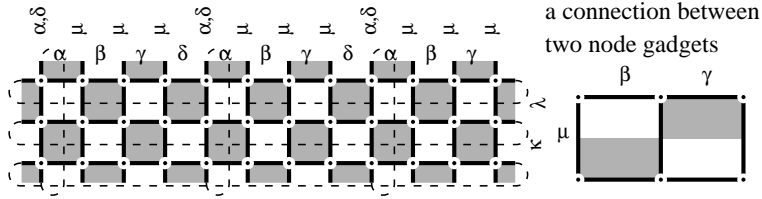,width=4in}}
\caption{\label{node-edge-gadget}
Node gadget $G_u$ for a node $u$ (left) and connections between two
node gadgets (right)
used in the proof of Theorem~\ref{main2}.
The dashed lines indicate wrap-around connections between boundary nodes
of the node gadget. The edge labels indicate the values (alleles)
in the second locus of each edge (individual). 
The wrap-around horizontal edges have label $\delta$. 
}
\end{figure}

We can view the second locus as labels on edges.
A one-letter label $a$ corresponds to a ``pair with a repeat'', \IE, $(a,a)$, and two-letter
label $a,b$ is a ``normal pair'' $(a,b)$.
Inside the $3\times 12$ grid of a node gadget the labels of
horizontal edges are equal if one edge is above another, and
in a $12$-edge ring of such edges labels repeat in a cycle of
$4$ (and each has one letter).
We have similar situation for vertical edges inside the grid.
The ``wrap-around'' edges (in every $4\tx$ column) are labeled
with proper pairs $\alpha,\delta$ such 
that they intersect the labels of their neighbors.
We assume that these labels are unique to every $G_u$ 
(in Figure~\ref{node-edge-gadget}, these would be labels $\delta_u$ and
$\alpha_u$).

The edges that connect node gadgets are labeled $\mu$ where
$\mu$ is the same in all node gadgets and the labels of
gadget edges that take part in the connection are the same in all gadgets
(thus $\beta$ and $\gamma$ are without implicit subscripts).

It is easy to see that every cycle of $4$ edges in our new graph
is indeed a full siblings set: according to the first locus they
are surely so 
and according to the second locus we can have only two distinct labels
on a cycle, \EG, $\{\alpha_u,\lambda\}$ or $\{\beta,\mu\}$. 
Edges with a ``normal pair'' 
label $\alpha,\delta$ do not belong to any
cycle of length $4$.

It is a bit more non-trivial to check that we have only
two types of full sibling sets of $3$ edges: subsets of
$4$-cycles, and sets with repeat label $\alpha$, repeat label $\delta$ and
normal label $\alpha,\delta$ that include ``wrap-around'' edges and
adjacent horizontal edges (one at each end).
Basically, if we have two horizontal edges from ``different
columns'' in a set, we cannot add any other label --- with
the exception we have just described.
Recall that a full sibling set of $3$ edges forms a path;
thus combination of labels like $\lambda$, $\delta$ and $\kappa$
is not full siblings.

We give each edge a {\em potential}. By default it is
equal to $0.25$.  The exceptions are: an edge with the label
$\alpha,\delta$ has a potential of $0.5$, an edge
with label $\mu$ 
that is not a center of a group of three
nodes in the node gadget that defined an edge connection
has a potential of $0.5$ and an edge with label $\mu$
that is a center of a group of three
nodes in the node gadget that defined an edge connection
has a potential of $0$.

By previous observations, no full siblings set has a potential exceeding
$1$. Note also that for each node of $G$ we distributed
a potential of $19.5$, so no cover with full siblings sets can use
fewer than $19.5\times 2n=39n=13m$ sets.

Assume that in $G$ we have a cut with $3n-c=m-c$ edges, \IE, a
partition of the set of nodes into $A$ and $B$ such that only
$c$ edges (of $m=3n$ edges) are inside the partitions.  We will show
a cover with $39n+c$ full siblings sets.  First we use cycles that
correspond to gray squares in every gadget $G_u$ such that
$u\in A$, and if $u\in B$ we use cycles that correspond to white
square.  This is $12$ sets per gadgets.  Next, in each gadget
we use $3$ triples centered on $\alpha,\delta$ edges.  Next,
in a connection between $A$ and $B$ we have either
two edges labeled $\beta$ already covered, or two
edges labeled $\gamma$: in the diagram, suppose that
the ``lower gadget'' is in $A$, then $\gamma$ is in
a gray square of that gadget; and as the upper gadget
is in $B$ and in that edge the upper $\gamma$ is covered
by a white cycle, it is already covered.  Thus we can
use a cycle with two $\beta$ edges and two $\mu$'s, and
one $\mu$ is left out.  This happens twice in a connection
between two gadget, so we add two cycles and one pair
of left-out $\mu$'s,  a total of $3$ sets.

If a connection is inside $A$ or inside $B$, then the
uncovered edges have one $\beta$ and one $\gamma$ and
they form a path of $5$ edges, which can be covered with
$2$ sets, and since this happens twice, we use $4$ sets.

Summarizing, we used $2n\times (12+3)+3n\times 3+c=39n+c$ sets.
This proves~{\bf (i')}.

Now, we prove~{\bf (ii')}. 
Suppose that we have a cover with $39n+c$ sets.
We have to normalize it so it will have the form of a
cover derived from a cut, without increasing the number
of sets.  The potential introduced above allows to make
local analysis during the normalization.  A set with
potential $p<1$ has a {\em penalty} of $1-p$, and we
have the sum of penalties equal $c$.

We can assign the penalty to node gadgets.  If a set with
a penalty is contained in some $G_u$ than the assignment
is clear.  If we have a set of two edges, then we assign
penalty of $0.25$ to each edge with potential $0.25$ and
if such an edge is contained in $G_u$, we assign the
penalty to $G_u$.

If $G_u$ has a penalty of $1$ or more, we remove $G_u$
from consideration and recursively normalize the cover
of the remaining gadgets.  Once we make this normalization,
we partition the remaining nodes into $A$ and $B$.  If
a node $u$ has at most one neighbor in $A$ we insert $u$ to
$A$, meaning, we cover it with gray cycles etc, and
we will add $19.5+1$ sets (an edge not covered counts as
half of a set, because we can combine them in pairs).

\begin{wrapfigure}[5]{r}{1.5in}
\vspace{-2ex}
\epsfig{file=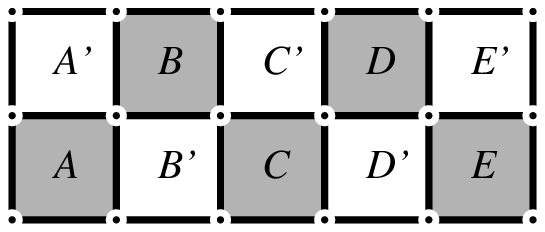,width=1.5in}
\end{wrapfigure}

Thus remains to normalize the cover of $G_u$ assuming
that its penalty is at most $0.75$.  Consider the
central horizontal cycle of the grid of $G_u$: it has
$12$ edges, and no two of them can belong to the same
full sibling set with more than $2$ edges; moreover, the
sets of at least $3$ edges to which they belong are fully
contained in $G_u$.  Because $G_u$ obtain at
most $0.75$ in penalties, at least $9$ edges of that $12$-cycle
are covered by full siblings $4$-cycles.  Consider the
longest connected fragment of such covered edges; assume
that they are covered with gray cycles.

Suppose that the last two cycles in that fragment are
$A$ and $B$ in the last diagram.  We want to change the
solution without increasing the number of set and use also
cycle $C$. If $C$ contains a set $S$ used in the current solution,
we can enlarge $S$ (making some other sets smaller) and our
fragment is extended.  If $C$ contains two edges contained
in two-edge sets, we can combine the sets so the latter two
are in one set, and again we can force $C$ into our solution.
So every edge of $C$ is in a different set from the current solution
and at most one of these sets is a pair.

Consider the edge on the boundary of $B'$ and $C$; if it is in
a set of more than $3$ edges, that set is contained in $C$ -- and
we excluded that case, or in $B'$ -- but only two edges of $B'$
remain uncovered.  Hence this edge is contained in a set with
two edges only, and it gets a penalty of $0.25$ that is delivered to $G_u$.

Consider the edge on the boundary of $C$ and $C'$.  According to our
case analysis, it is contained in a set of at least $3$ edges, and
which has only one edge in $C$, so this set is contained in $C'$.
Because $A$ covers one edge of $C'$, we have a set of exactly $3$
edges that gets a penalty of $0.25$, and thus $G_u$ already got $0.5$
of penalty.

We repeat the same reasoning at the other end of the fragment and
we double the penalty to $1$. The only doubt we can have is that
we are counting one of the penalties twice.  But this is not
possible: the other end of the fragment cannot be covered by $C$,
and it cannot be covered by $D$, as we use the set $C'\setminus B$
which overlaps $D$.  If the other end of our fragment is covered with
$E$, then we get penalties for the boundary of $D$ and $D'$, and
for the set $D'\setminus E$ and we have no double counting. Other cases are
similar.

Now an explicitly normalized node gadget has a center row covered
with $12$ cycles of the same color. The wrap-around edges with
$\alpha,\delta$ labels can be included in paths of $3$ edges -- and with
potential $1$; note that after we committed ourselves to $12$ ``central''
cycles, the edges of such a path do not belong to any other set with
more than two edges.  Now the uncovered edges are only in the
connection gadgets and they form sets of $5$ edges, with no connections
between them.  We have two such $5$-tuples for each connection.

We split the nodes according to the colors used in their gadgets:
gray cycles are in set $A$ and white cycles are in set $B$.  If we have a
$5$ tuple of an $A-B$ connection, its uncovered edges form a cycle
and an edge, so we can cover it with $1.5$ sets and we cannot do
any better.  If we have an $A-A$ or $B-B$ connections, the
uncovered edges form a path of $5$ edges and we much cover them
with two sets.

This completes the hardness reduction.

On the algorithmic side, we can use the result of 
Berman and Krysta~\cite{BerK03}. 
For polynomial time, we have to round the rescaled weights to small
integers, so the approximation ratio should have some $\epsilon$ added.
The $2$-IMP with rescaled weight has an 
approximation ratio of $\beta a$, where
for $a=3$ $\beta=2/3$, 
for $a=4$ $\beta=0.6514$
and
for $a>4$ $\beta=0.6454$. 
We can greedily find a maximal packing with sets of size $4$ and find $1/2$ of
the remaining sets of size $3$ using $2$-IMP algorithm of~\cite{BerK03}.  
Easy analysis shows that that this gives an approximation ratio of $3/2$.
\end{proof}

\begin{remark}
Using a reduction again from $3$-MAX-CUT that is similar in flavor
to the above proof (but with different gadgets, different covering
components and simpler case analysis) one can prove that,  
assuming RP$\neq$NP,
there is no $((1182/1181)-\eps)$-approximation algorithm for \All$_{n,\ell}$
even if $a=6$ and $\ell=O(n)$
for any constant $\varepsilon>0$.
\end{remark}

\section{Inapproximability for \All$_{n,\ell}$ and \TAll$_{n,\ell}$ for $a=n^\delta$} 

\begin{Lemma}\label{high-inapprox}
For any two constants $0<\eps<\delta<1$ with $a=n^\delta$,
\All$_{n,\ell}$ and 
\TAll$_{n,\ell}$ are
$n^\eps$-inapproximable 
assuming NP$\not\subseteq$ZPP.
\end{Lemma}

\begin{proof}
For any two constants $0<\eps<\delta<1$, 
consider a hard instance $G=(V,E)$ of the graph coloring problem with $n$
vertices $[n]=\{1,2,\ldots,n\}$
and $\Delta^\ast(G)\leq |V|^{\delta}$. As observed in the proof of Theorem~\ref{basic-approx}, 
it will be sufficient to translate this to an instance $\cal J$ of the $2$-label cover 
problem.
We will have a individual
for every vertex $i$. We will translate an edge $\{i,j\}\in E$ to
{\em exactly} $n-2$ ``forbidden triplets'' of individuals 
$\{\,\{i,j,k\}\,|\,k\in [n]\setminus\{i,j\}\}$ 
of the $2$-label cover problem such that each of these 
set of individuals cannot be a full sibling group. We call 
$\{i,j\}$ as the ``anchor'' of these triplets.
The translation is done by 
by introducing
a new locus and three labels $a$, $b$ and $c$,  
putting $a$ and $b$ as the labels of individuals $i$ and $j$ in this locus,
and putting $c$ as the label of every other individual in this locus.
Finally, we use the following 
distinctness gadgets, if necessary, to
ensure that all the individuals are distinct.
There are at most $O(n^2)$ such gadgets. 
The purpose of such gadgets is to make sure no two individuals are
identical, \IE, every pair of individuals differ in at least one locus,
while still allowing any subset of individuals to be in a full
sibling group.
Consider a pair of individuals $u$ and $v$ 
that have the same set of loci.
Select a {\em new} locus, two 
symbols, say $a$ and $b$, 
and put $a$ in the locus of all individuals except $v$ and put $b$ 
in the locus of $v$. 

It suffices to show that our reduction has the following properties:
\begin{description}
\item[(1)]
A set of $x\geq 3$ vertices of $G$ are independent if and only if
the corresponding set of $x$ individuals in $\cal J$ is a valid
full sibling group.

\item[(2)] 
If $G$ can be colored with $k$ colors then 
$\cal J$ can be covered with $k$ sibling groups.

\item[(3)] 
If $\cal J$ can be covered with $k'$ sibling groups then 
$G$ can be colored with no more than $2k$ colors.
\end{description}
Suppose that we have a set $S$ of independent vertices in $G$.
Suppose that the corresponding set of individuals in $\cal J$ cannot
be a full sibling group and thus must include a forbidden triplet
$\{i,j,k\}$ with $\{i,j\}$ as the anchor. Then $\{i,j\}\in E$,
thus $S$ is not an independent set. 
Conversely, 
suppose that the set of individuals $\cal J$ is
be a full sibling group. Then, they cannot include a forbidden triplet.
This verifies Property~{\bf (1)}. 

Suppose that $G$ can be colored with $k$ colors. We claim that the set of individuals corresponding
to the set of vertices with the same color constitute a sibling group for either problem.
Indeed, since the set of vertices of $G$ with the same color are
mutually non-adjacent, they do not include a forbidden triplet.
This verifies Property~{\bf (2)}. 

Finally, suppose that the instance of the generated 
$2$-label cover problem has a solution with $k'$
sibling groups. For each sibling group, select a new color and assign it to all the individuals
in the group. Now, map the color of individuals in $\cal J$
to the corresponding vertices of $G=(V,E)$. 
Let $E'\subseteq E$ be the set of edges which connect two vertices of the same color.
Note that in the graph $G'=(V,E')$ every vertex is of degree at most one since otherwise
the sibling group that contains these three individuals corresponding to 
the three vertices that comprise the two adjacent edges has a forbidden triplet.
Thus, we can color the vertices of $G'$ from a set $C$ of two colors.
Obviously, the graph $G''=(V,E\setminus E')$ can be colored with colors from a set $D$
of $k'$ colors.
Now, it is easy to see that $G$ can be colored with at most $k\leq 2k'$ colors: 
assign a new color to every pair in $C\times D$ and color a vertex with the color
$(c,d)\in C\times D$ where $c$ and $d$ are the colors that the vertex received in 
the coloring of $G'$ and $G''$, respectively. 
This verifies Property~{\bf (3)}. 
\end{proof}

\section{Approximating Maximum Profit Coverage (MPC)}

\begin{Lemma}\label{MPC-a-b}~\\
\noindent
{\bf (a)} 
MPC is NP-hard for $a\geq 3$ and $a^c$-inapproximable 
for arbitrary $a$ 
and some constant $0<c<1$
assuming P$\neq$NP 
even if every set has weight $a-1$, every element has weight $1$ and 
every set contains exactly $a$ elements.
The hard instances can further be restricted such that 
each element is a point in some underlying metric space and each set correspond to a ball of
radius $\alpha$ for some fixed specified $\alpha$.

\noindent
{\bf (b)} 
MPC is polynomial-time solvable for $a\leq 2$. Otherwise, 
for any constant $\eps>0$, 
MPC admits $(0.5a+0.5+\eps)$-approximation for fixed $a$ and
$(0.6454 a+\eps)$-approximation otherwise.
\end{Lemma}

\begin{proof}~\\
\noindent
{\bf (a)} 
Consider an instance of the independent set problem on a $a$-regular graph 
$G=(V,E)$. Build the following instance of the MPC problem. The universe $U$
is $E$. For every vertex $v\in V$, there is a set $S_v$ consisting of the edges incident
on $v$. Finally, set the weight of every element to be $1$ and the weight of every set to be $a-1$. 
Note that each set contains exactly $a$ elements. 

It is clear that an independent set of $x$ vertices correspond to a solution of the MPC problem of
profit $x$ by taking the sets corresponding to the vertices in the solution. Conversely, 
suppose that a solution of the MPC problem contains two sets $S$ and $S'$ that have a non-empty
intersection. Since each set contains exactly $a$ elements, removing one of the two sets from the
solution does not decrease the total profit. Thus, one may assume that
every pair of sets in a solution of the MPC problem has empty intersection. Then, such a solution 
involving $x$ sets of total profit $x$ correspond to an independent set of $x$ vertices.

If one desires, one can further restrict the instance of the MPC problem in {\bf (a)} above to the case where 
each element is a point in some underlying metric space and each set correspond to a ball of
radius $\alpha$ for some fixed specified $\alpha$. 
All one needs to do is to use the standard trick
of setting the weight of each edge in the graph to be $\alpha$ and define the distance between two
vertices to be the length of the shortest path between them.

\noindent
{\bf (b)} 
Consider the weighted set-packing problem and 
let $a$ denote the maximum size of any set. 
For fixed $a$, it is easy to use the algorithm for the weighted set-packing as a black box
to design a $a/2$-approximation for the MPC problem. For each set $S_i$ of MPC, consider
all possible subsets of $S_i$ and set the weight $w(P)$ of each subset $P$ to be the sum of weights of 
its elements minus $q_i$. Remove any subset from consideration if its weight is negative. The collection
of all the remaining subsets for all $S_i$'s form the instance of the weighted set-packing problem.

It is clear that a solution of the weighted set-packing will never contain two sets 
$S$ and $S'$ that are subsets of some $S_i$ since then the solution can be improved by 
removing the sets $S$ and $S'$ and adding the set $S\cup S'$ to the solution (the solution cannot
contain the set $S\cup S'$ because of the disjointness of sets in the solution). Thus, at most
one subset of any $S_i$ is used the solution of the weighted set-packing. If a subset
$S$ of some $S_i$ was used, we use the set $S_i$ in the solution of the MPC problem;
note that the elements in $S_i\setminus S$ must be covered in the solution by other sets
since otherwise there is a trivial local improvement.
In this way, a solution of the weighted set-packing of total weight $x$
corresponds to a solution of the MPC problem of total profit $x$. 
Conversely, in an obvious manner a solution of the MPC problem of total profit $x$ 
corresponds to solution of the weighted set-packing of total weight $x$.

For $a\leq 2$, weighted set-packing can be solved in polynomial time via 
maximum perfect matching in graphs.

For fixed $a>2$, Berman~\cite{B00} provided an approximation algorithm based on local improvements
for this problem produces an approximation ratio of $\frac{a+1}{2}+\eps$ for any constant $\eps>0$. 
An examination of the algorithm in~\cite{B00} shows that 
the running time of the procedure for our case is $O\left(2^{(a+1)^2} m^{a+1}\right)=O(m^{a+1})$. 

When $a$ is {\em not} a constant, 
Algorithm $2$-IMP of Berman and Krysta~\cite{BerK03} can be adapted for MPC to run in polynomial time.
For polynomial time, we have to round the rescaled weights to small
integers, so the approximation ratio should have some $\epsilon$ added.
The $2$-IMP with rescaled weight has an 
approximation ratio of $0.6454 a$ for any $a>4$. 
However, we need a somewhat complicated dynamic programming procedure
to implicitly maintain all the subsets
for each $S_i$ without explicit enumeration. 

Here are the technical details of the adaptation.
We will view sets that we can use as having {\em names} and elements.
A name of $A$ is a set $N(A)$ given in the problem instance, and elements form
a subset $S(A)\subset N(A)$.  The profit $w(S)$ is sum of weights of elements
minus the cost of the naming set, $p(A)=w(S(A))-c(N(A))$.

The algorithm attempts to insert two sets to the current packing and
remove all sets that overlap them; this attempt is successful if the
sum of weights raised to power $\alpha >1$ increases; more precisely,
the increase should be larger then some $\delta$, chosen is such a way
that it is impossible to perform more than some polynomial time of
successful attempts.  As a result, we can measure the weights of sets
with a limited precision, so we have a polynomially many different
possible weights.

When we insert set with name $B$ that overlaps a set $A$ currently
in the solution, we have a choice: remove set $A$ from the solution
or remove $A\cap B$ from $B$.  If we also insert a set with name $C$
we have the same dilemma for $A$ and $C$.  Our choice should maximize
the resulting sum of $w^\alpha(S)$ for $S$ in the solution.  

If we deal with two sets, we can define the quantities

\noindent
$x_A=p(A-B)$\\
$x_B=p(B-A)$\\
$w_{AB}=w(A\cap B)$.

If we include $A\cap B$ in $A$, the modified profit is
$(x_A+w_{AB})^\alpha+x_B^\alpha$.

If we include $A\cap B$ in $B$, and remove $A$, the modified profit is
$(x_B+w_{AB})^\alpha$.

Our problem is that we know 
$y_1=x_A^\alpha$ and
$y_1=w_{AB}$
but we do not know $x_B$, because the exact composition of $B$ depends on many decisions.
Thus we do not know if the following inequality holds for $x=x_B+x_{AB}$:

$$(y_1+y_2)^+(x-y_2)^\alpha \le x^\alpha .$$

It is easy to see that the left-hand-side grows slower than the right-hand side,
so once the inequality holds, it is true for all larger $x$.  For this reason it is
never optimal to split $A\cap B$ between the two sets, instead we allocate the
overlap to one of them.

The situation is similar when we insert two sets.  To decide how to handle each overlap
of the (names of) sets that we are inserting with the sets already in the solution, it
suffices to know their profits.  Because we measure profits with a bounded precision,
we can make every possible assumption about these two profits, make the decisions and
check if the resulting profits are consistent with the assumption; if not, we ignore
that assumptions.  Among assumptions that we do not ignore, we select one with the
largest increase of profits raised to power $\alpha$.  If one of them is positive, we
perform the insertion.

Thus we can select a pair of insertion in polynomial time even though we have
a number of candidates that is proportional to $n2^a$.  Thus our algorithm runs in
polynomial time even for $a >> \log n$.  Therefore we can achieve the approximation
ratio of $2$-IMP, \IE, $0.6454a+\eps$, which is better than factor $a$ offered by a greedy
algorithm: keep inserting a set with maximum profit that does not overlap an already
selected set.
\end{proof}

\section{Approximating $2$-coverage}

\begin{Lemma}\label{DS}~\\
\noindent
{\bf (a)} 
For $f=2$, 
$2$-coverage is $(1+\eps)$-inapproximable
for some constant $\eps>0$ unless \\
NP$\not\subseteq \cap_{\eps>0}$BPTIME$(2^{n^\eps})$
and admits $O(m^{\frac{1}{3}-\eps'})$-approximation for some constant $\eps'>0$.

\noindent
{\bf (b)} 
For arbitrary $f$, $2$-coverage admits $O(\sqrt{m})$-approximation.  
\end{Lemma}

\begin{proof}~\\
\noindent
{\bf (a)}
Consider an instance $<G,k>$ of the densest subgraph problem. 
Then, define an instance of the $(k,2)$-coverage problem such that 
$U=E$, there is a set for every vertex in $V$ that contains all the edges
incident to that vertex, and we need to pick $k$ sets. Note that 
for this instance $f=2$.

For the other direction, define a vertex for every set, connect two vertices
if they have a non-empty intersection with a weight equal to the number of
common elements. This gives an instance of {\em weighted} DS whose goal
is to maximize the sum of weights of edges in the induced subgraph and
admits a $O(m^{\frac{1}{3}-\eps})$-approximation for some constant $\eps>0$~\cite{FPK01}. 

\noindent
{\bf (b)}
For notational convenience it will be convenient to define the $(k,\ell)$-coverage
problem (for $\ell\geq 1$) which is same as the $2$-coverage problem with $k$ sets
to be selected except that every element must belong to at least $\ell$ selected sets
(instead of two selected sets). 
We will also use the following notations.  
OPT$(k,\ell,\cS)$ is the maximum value of the objective function for
the $(k,\ell)$-coverage problem on the collection of sets in $\cS$ and 
A$(k,\ell, \cS)$ is the value of the objective function for
the $(k,\ell)$-coverage problem on the collection of sets in $\cS$
computed by our algorithm. 
For notational convenience, let $\wp=1-(1/\bee)$.
We will give both an $O(k)$ and an $O(m/k)$ approximation which together
gives the desired approximation.

The following gives an $O(k)$-approximation. 
Create a new set $T_{i,j}=S_i\cap S_j$ for every pair of indices $i\neq j$.
Run the $(k/2,1)$-coverage $\wp$-approximation algorithm on the
$T_{i,j}$'s and output the elements and, for each selected $T_{i,j}$,
the corresponding $S_i$ and $S_j$. Note that each element is covered at least
twice. One can look at all the $k\choose 2$ pairwise intersections of 
sets in an optimal solution of $(k,2)$-coverage on $\cS$, consider the
$k/2$ pairs that have the largest intersections and thus conclude that
an optimal solution of $2$-coverage on $\cS$ covers no more than 
$O(k)$ times the number of elements in an optimal solution of the
$(k/2,1)$-coverage on the $T_{i,j}$'s. 

To get an $O(m/k)$-approximation, first note that
OPT$((k/2),1,\cS)\geq$ OPT$(k,2,\cS)$. Run the $\wp$-approximation
algorithm to select the collection of sets $\cT\subseteq\cS$ to approximate 
OPT$((k/2),1,\cS)$. 
For each remaining set in $\cS\setminus\cT$, remove all elements that do not belong to
the sets in $\cT$ and remove all elements that are already covered twice in $\cT$. 
We know that if we were allowed to choose all of the
$m-k$ remaining sets in $\cS\setminus\cT$ we would cover all the elements in the sets $\cT$.
But since we are allowed to choose only additional $k/2$ sets, we 
choose those $k/2$ sets from $\cS\setminus\cT$ that cover the maximum number of elements 
in the union of sets in $\cT$. This involves again running the $\wp$-approximation 
algorithm. We will cover at least a fraction $k/(2m)$ of the maximum number of elements.
\end{proof} 

\section{Conclusion and Further Research}

In this paper we investigated four covering/packing problems
that have applications to several problems in bioinformatics.
Several questions remain open on the theoretical side. For example,
can stronger inapproximability 
results be proved for \All$_{n,\ell}$ and \TAll$_{n,\ell}$
intermediate values of $a$ and $\ell$ that are excluded in our proofs? 

\vspace*{0.2in}
\noindent
{\bf Acknowledgments}

We would like to thank the anonymous reviewers for their helpful comments
that led to significant improvements in the presentation of the paper.

\end{document}